\documentclass[psamsfonts]{amsart}
\usepackage{amssymb}
\usepackage{graphicx}
\usepackage[tight]{subfigure}
\usepackage{float}

\usepackage{amssymb,amsfonts}
\usepackage[all,arc]{xy}
\usepackage{enumerate}
\usepackage{mathrsfs}
\usepackage{booktabs}
\usepackage{algorithmic}

\newtheorem{thm}{Theorem}[section]
\newtheorem{cor}[thm]{Corollary}
\newtheorem{prop}[thm]{Proposition}
\newtheorem{lem}[thm]{Lemma}

\newtheorem{remark}[thm]{Remark}

\theoremstyle{definition}

\theoremstyle{remark}

\numberwithin{equation}{section}

\title[A simple and efficient method for option pricing]{A simple and efficient numerical method for pricing discretely monitored early-exercise options}\thanks{Disclaimer: The statements made and opinions expressed here are solely our own, and do not reflect the views of China Mechants Bank or its employees and affiliates.}

\author{Min Huang and Guo Luo}

\begin{document}

\begin{abstract}
We present a simple, fast, and accurate method for pricing a variety of discretely monitored options in the Black-Scholes framework, including autocallable structured products, single and double barrier options, and Bermudan options. The method is based on a quadrature technique, and it employs only elementary calculations and a fixed one-dimensional uniform grid. The convergence rate is $O(1/N^4)$ and the complexity is $O(MN\log N)$, where $N$ is the number of grid points and $M$ is the number of observation dates.
\end{abstract}
\maketitle

%\tableofcontents

\section {\large Introduction}
\label{sec_intro}

Exotic options are commonly traded throughout the world. Many popular exotic options are path-dependent and have early-exercise features. These options can often be priced using analytical formulas if they are continuously monitored (e.g. barrier options). In practice, however, most path-dependent exotic options are discretely monitored \cite{cont1}, in which case they need to be priced using numerical techniques. Due to the complicated structures of these options, traditional pricing models based on Monte-Carlo simulations and finite difference methods are often too time-consuming to be useful in practical situations. More recent pricing methods based on advanced mathematical techniques, on the other hand, tend to be more efficient (e.g. \cite{barclose,hilbert,cos,berm2}), but for many financial institutions, these methods are often too difficult to understand and to properly implement. To strike a balance between model performance and practical utility, we propose a new quadrature-based method that is much faster and more accurate than the traditional Monte-Carlo and PDE methods, yet at the same time is easy to understand and to implement. We will first give a brief review of the types of products considered, as well as the quadrature-based pricing model which is the foundation of our work. Then we will explain our method and provide numerical examples.

\subsection{Autocallable Structured Products}

Autocallable structured products belong to a class of exotic options with early-exercise features. Many different types of autocallable products have been created and traded in financial markets, and they have become increasingly popular in recent years. We refer to the Appendix of \cite{autopde} for a description of the main features of various autocallable products.

We will consider a very common autocallable product with discrete observation dates. At each observation date there is a pre-specified barrier level. If the price of the underlying asset is greater (less) than or equal to the barrier level (depending on the terms of the product), the option is exercised and a pre-specified fixed-rate return is paid. If the asset price is below (above) the barriers at all observation dates, the option is never exercised and the investor receives a negative return at maturity. In addition, autocallable products may have a knock-in feature. In this case, if the option is never exercised, the negative return the investor receives depends on whether the asset price at maturity reaches a pre-specified knock-in level. While the value of a continuously monitored autocallable product has a simple closed-form solution, the value of a discretely monitored autocallable product cannot be calculated easily. In the discrete case, there exist analytical solutions in terms of multiple integrals, cf. \cite{autoclose1,autoclose2,autopde}. The numerical calculation of these integrals, however, can become prohibitive if the number of observation dates exceeds five. In practice, discretely monitored autocallable products are commonly priced using Monte-Carlo simulations. This method is straightforward, but convergence is usually slow and acceleration techniques such as variance reduction are often needed (cf. \cite{automc1,automc2}). Another popular method for pricing discretely monitored autocallable products is to solve the governing Black-Scholes partial differential equation (PDE) using finite difference method (cf. \cite{autopde}). Assuming a second-order central difference approximation in space, the overall convergence rate of a typical finite difference based pricing method is $O((\delta x)^2 + \delta t)$ if the explicit forward Euler method is used in time, and $O((\delta x)^2 + (\delta t)^2)$ if the implicit Crank-Nicolson method is used. Since two-dimensional grids are needed for finite difference methods, computational complexities are at least of order $1/(\delta x \delta t)$. In addition, since the payoffs of autocallable products are discontinuous (in asset price), additional care (such as smoothing of payoff functions) must be taken to ensure the accuracy of any finite difference approximations.

\begin{remark}
Some autocallable products may have payoffs at maturity that are of the same type as that of European vanilla options. These products can be effectively viewed as combinations of autocallable products and barrier options (see below).
\end{remark}

\subsection{Discrete Barrier Options}

Barrier options are among the most popular types of exotic options. A barrier option may be activated (knock-in option) or deactivated (knock-out option) when the price of the underlying asset crosses certain barrier levels. Barrier options may be discretely or continuously monitored. A single barrier option has one barrier at each observation date, while a double barrier option has two barriers at each observation date. The final payoff of a barrier option (if it is active at maturity) may be of the same type as that of a vanilla option or that of a digital option. A special type of barrier option has a constant amount of cash as the final payoff, and such an option is called a touch option. Most barrier options have time-independent barrier levels, but options with time-dependent barrier levels have also been studied \cite{time}.

Similar to the case of autocallable structured products, there are closed-form solutions for discrete barrier options in terms of multiple integrals \cite{conv1}, but such solutions are often difficult to evaluate directly. In practice, discrete barrier options can be priced using Monte-Carlo simulations or standard binomial tree methods, but these methods are usually slow \cite{cont2}. Other methods that have been proposed to price discrete barrier options include continuity correction approximations \cite{cont1,interp}, Wiener-Hopf methods \cite{barclose}, adaptive mesh methods \cite{amm}, Hilbert transform methods \cite{hilbert}, finite element methods \cite{ele}, Fourier-cosine series expansion methods \cite{cos}, and quadrature methods \cite{quad,gauss}. These methods, while useful in certain contexts, have not been as widely used as the traditional Monte-Carlo and finite difference methods, usually due to their complexity.

\subsection{Bermudan Options}

Bermudan options are discrete versions of American options. A Bermudan option can be exercised at any of the prescribed observation dates, and the payoff is of the same type as that of a vanilla option. Similar to discrete barrier options, Bermudan options can be priced using Monte-Carlo simulations \cite{berm1}, Hilbert transform methods \cite{berm2}, Fourier-cosine series expansion methods \cite{cos}, and quadrature methods \cite{fft,conv,gauss}.

\subsection{Overview of the Quadrature Method}

Among the various methods proposed to price discretely monitored options, the quadrature method is particularly appealing because of its high efficiency and accuracy. The method has been applied to discrete barrier options and Bermudan options \cite{quad,fft,conv,gauss}. The main idea is to solve for option values at each observation date via backward induction in time. The risk-neutral valuation formula is expressed as a single integral, which is then evaluated numerically to produce the option price. Specifically, let $V$ denote the value of the option, $S$ the value of the underlying asset, $r$ the risk-neutral interest rate, $t_1,\dotsc,t_M$ the observation dates, and $\mathbb{E}$ the risk-neutral expectation. If the underlying asset $S$ does not trigger an early exercise at $t_m$, we have
\begin{align*}
  V(t_m,S) & = e^{-r(t_{m+1}-t_m)} \mathbb{E} \bigl[ V(t_{m+1},\cdot)|S \bigr] \\
  & = e^{-r(t_{m+1}-t_m)} \int_0^\infty V(t_{m+1},y) f(y|S)\,dy,
\end{align*}
where $f(y|S)$ is a probability density function whose form depends on the model of the underlying asset. If $S$ triggers an early exercise at $t_m$, on the other hand, the option price $V(t_m,S)$ would be equal to a prescribed value. The integral above can be calculated using FFT \cite{conv,fft} or Fast Gauss Transform \cite{gauss}. Since, however, $V(t_{m+1},y)$ is discontinuous (in $y$) for autocallable products and barrier options, and non-differentiable (in $y$) for Bermudan options, care must be taken to ensure the accuracy of the numerical evaluation of the integral. While several shifted or nonuniform grids have been designed in previous studies to address this difficulty \cite{conv,gauss}, the problem becomes particularly challenging when multiple discontinuities are present at each observation point, for instance in the case of double barrier options with time-dependent barrier levels.

\subsection{Motivation of Our Work}

Although a good number of advanced techniques have been proposed to improve pricing models' accuracy and efficiency, in most practical situations, simple methods that are more cost effective are usually preferred over their more sophisticated counterparts. The reason for this is twofold. First of all, when data quality is not high enough, sophisticated models are not necessarily beneficial. For instance, interest rates and volatilities are crucial components of nearly every pricing model, but they need to be estimated from available market data. If the estimated parameters contain large errors, which is not uncommon in products sold in emerging markets, any advantages gained from the use of sophisticated models may be (more than) offset by these errors, making the simpler models more attractive. Secondly, implementations of pricing models usually involve staff members from multiple business departments, and the resulting products often need active maintenance and updates. As a result, models that are too complicated in nature may hinder effective business communications, which increases maintenance costs and operational risks. In view of these considerations, it is not difficult to see why traditional Monte-Carlo and PDE methods are still among the most popular methods in the valuation of discretely monitored options, even though their computational costs are already high enough to adversely impact their applicability in business.

In view of the practical concerns mentioned above, we propose a new quadrature method to price the aforementioned discretely monitored options in the Black-Scholes framework. The convergence rate of our method is $O(1/N^4)$ and the complexity is $O(MN\log N)$, where $N$ is the number of grid points and $M$ is the number of observation dates. The performance of our method is on par with previous quadrature-based methods such as the CONV method \cite{conv}, but it is more straightforward, and is better suited for products with multiple discontinuities. Our method differs from other quadrature methods mainly in three aspects. First, we work with probability density functions directly instead of using characteristic functions or Toeplitz matrices. Secondly, we use only a fixed one-dimensional uniform grid to compute all integrals. Thirdly, we utilize explicit Black-Scholes formulas to improve the accuracy of the calculations. Due to these novel modifications, our method is very easy to implement, and is capable of handling sophisticated products such as double-barrier options with time-dependent barrier levels.

\subsection{Organization of the Paper}

The rest of the paper is organized as follows. Section \ref{sec_assump} specifies the class of (discrete) option pricing problems that our quadrature method is intended to solve, and Section \ref{sec_outl} presents the main recursion formula for our method. After detailing the implementation of our method in Section \ref{sec_impl}--\ref{sec_berms}, we summarize the algorithm in Section \ref{sec_algo} and then present numerical examples in Section \ref{sec_exmp}. The Appendix collects a few useful theoretical results, which lay the foundation of a class of (discrete) option pricing algorithms (including the one described here) but which, to our best knowledge, do not seem to have been rigorously proved or even properly formulated in the literature. We supply the proofs here in the hope that they would be useful to interested readers.

\section{Basic Assumptions}
\label{sec_assump}

We assume that the price of the underlying asset $S(t)$ satisfies the following stochastic differential equation in the risk-neutral measure:
\begin{equation}
  \label{sde}
  dS(t) = \bigl[ r(t)-q(t) \bigr] S(t)\,dt + \sigma(t) S(t)\,dW(t),
\end{equation}
where $r(t)$ is the risk-neutral interest rate, $q(t)$ is the yield rate, $\sigma(t)$ is the volatility, and $W(t)$ is the Wiener process.

In practice, interest rates are always time dependent. Yield rates for FX products are simply foreign interest rates, and for other types of products they may be implied from futures prices. Thus yield rates are usually time dependent as well. Implied volatilities are time dependent, whereas historical volatilities can often be taken as constant.

The solution to \eqref{sde} is
\begin{equation}
  \label{st}
  S(t) = S(t_0) \exp \biggl\{ \int_{t_0}^{t} \bigl[ r(s)-q(s) - \tfrac{1}{2} \sigma^2(s) \bigr]\,ds + \int_{t_0}^{t} \sigma(s)\,dW(s) \biggr\},
\end{equation}
where $t_0$ is the present date. We consider a discretely monitored option with observation dates $t_1,\dotsc,t_M$, where the last observation date $t_M$ is the maturity date. It follows from \eqref{st} that each $S(t_m)$ for $1 \leq m \leq M$ has the lognormal distribution
\begin{equation}
  \label{stdis}
  S(t_m) \sim S(t_0) \exp \biggl\{ \int_{t_0}^{t_m} \bigl[ r(s)-q(s) - \tfrac{1}{2} \sigma^2(s) \bigr]\,ds + \Bigl( \int_{t_0}^{t_m} \sigma^2(s)\,ds \Bigr)^{1/2}\,Z \biggr\},
\end{equation}
where $Z$ denotes the standard normal distribution. Now define
\[
  r_m = \int_{t_{m-1}}^{t_m} \frac{r(s)}{\Delta t_m}\,ds,\quad
  q_m = \int_{t_{m-1}}^{t_m} \frac{q(s)}{\Delta t_m}\,ds,\quad
  \sigma_m^2 = \int_{t_{m-1}}^{t_m} \frac{\sigma^2(s)}{\Delta t_m}\,ds,
\]
for $1 \leq m \leq M$ where $\Delta t_m = t_m-t_{m-1}$, and define piecewise constant functions
\[
  \tilde{r}(t) = r_m,\quad \tilde{q}(t) = q_m,\quad \text{and}\quad \tilde{\sigma}(t) = \sigma_m,\qquad \text{for}\qquad t_{m-1} < t \leq t_m.
\]
For the process
\[
  d\tilde{S}(t) = \bigl[ \tilde{r}(t)-\tilde{q}(t) \bigr] \tilde{S}(t)\,dt + \tilde{\sigma}(t) \tilde{S}(t)\,dW(t),
\]
since
\begin{align*}
  \int_{t_0}^{t_m} r(s)\,ds & = \sum_{n=1}^{m} \int_{t_{n-1}}^{t_n} r(s)\,ds
    = \sum_{n=1}^{m} r_n \Delta t_n = \int_{t_0}^{t_m} \tilde{r}(s)\,ds, \\
  \int_{t_0}^{t_m} q(s)\,ds & = \sum_{n=1}^{m} \int_{t_{n-1}}^{t_n} q(s)\,ds
    = \sum_{n=1}^{m} q_n \Delta t_n = \int_{t_0}^{t_m} \tilde{q}(s)\,ds, \\
  \intertext{and}
  \int_{t_0}^{t_m} \sigma^2(s)\,ds & = \sum_{n=1}^{m} \int_{t_{n-1}}^{t_n} \sigma^2(s)\,ds
    = \sum_{n=1}^{m} \sigma_n^2 \Delta t_n = \int_{t_0}^{t_m} \tilde{\sigma}^2(s)\,ds,
\end{align*}
it follows that $\tilde{S}(t_m)$ has the same distribution as $S(t_m)$ in \eqref{stdis} for each $m$. Since the value of the option depends only on probability distributions of the asset price at observation dates, the option value remains the same if we replace the process $S$ by the process $\tilde{S}$. In other words, we may safely assume that $r(t),\ q(t)$, and $\sigma(t)$ are piecewise constant functions. Thus in what follows, we shall assume
\[
  r(t) = r_m,\quad q(t) = q_m,\quad \text{and}\quad \sigma(t) = \sigma_m,\qquad \text{for}\qquad t_{m-1} < t \leq t_m.
\]

Consider now a general class of discretely monitored options with barriers. Since the sum of a knock-in barrier option and a knock-out barrier option with the same observation dates and barrier levels is a vanilla option (or a digital option if the barrier options are digital), to study the pricing of these discretely monitored options, it suffices to consider a knock-out barrier option which ceases to exist when barrier levels are crossed. To this end, assume that
{\renewcommand{\labelenumi}{(\Alph{enumi})}
\begin{enumerate}
\item {The option has two strike prices $K_m^-,\, K_m^+ \in [0,\infty]$, with $K_m^- \leq K_m^+$, at each observation date $t_m,\ m = 1,2,\dotsc,M$.}
\item {The option is exercised if $S \leq K_m^-$ or $S \geq K_m^+$ at some $t_m$, and the payoffs are given by $a_m^- S + b_m^-$ (if $S \leq K_m^-$) and $a_m^+ S + b_m^+$ (if $S \geq K_m^+$), respectively, for some $a_m^\pm,\, b_m^\pm \in \mathbb{R}$.}
\item {The final payoff at maturity is
\[
  V(t_M,S) = a_M S + b_M,\qquad \text{for}\qquad K_M^- < S < K_M^+.
\]}%
\end{enumerate}}%
These assumptions are general enough to cover a wide class of discretely monitored options, such as the ones mentioned in the introduction. For instance, common up-and-out autocallable products would have
\begin{align*}
  1 \leq m \leq M:&\quad 0 < K_m^+ < \infty,\ K_m^- = 0,\ a_m^+ = 0,\ b_m^+ > 0; \\
  m = M:&\quad a_M = 0,\ b_M < 0.
\end{align*}
Down-and-out put barrier options would have
\begin{align*}
  1 \leq m \leq M-1:&\quad K_m^+ = \infty,\ 0 < K_m^- < \infty,\ a_m^- = 0,\ b_m^- = 0; \\
  m = M:&\quad K_M^- = 0,\ 0 < K_M^+ < \infty,\ a_M = -1,\ b_M = K_M^+, \\
  &\quad a_M^+ = b_M^+ = 0.
\end{align*}
Double barrier knock-out call options would have
\begin{align*}
  1 \leq m \leq M-1:&\quad 0 < K_m^\pm < \infty,\ a_m^\pm = 0,\ b_m^\pm = 0; \\
  m = M:&\quad K_M^+ = \infty,\ 0 < K_M^- <\infty,\ a_M = 1,\ b_M = -K_M^-, \\
  &\quad a_M^- = b_M^- = 0.
\end{align*}
Bermudan put options with strike $K$ would have
\begin{align*}
  1 \leq m \leq M:&\quad K_m^-:\ \text{the unique solution of}\ K - K_m^- = V(t_m,K_m^-), \\
  &\quad K_m^+ = \infty,\ a_m^- = -1,\ b_m^- = K; \\
  m = M:&\quad a_M = 0,\ b_M = 0.
\end{align*}
We will give a proof of the uniqueness of $K_m^\pm$ for Bermudan options in the Appendix.

To summarize, our basic assumptions are
\begin{enumerate}
\item {The underlying asset price $S$ follows a geometric Brownian motion with piecewise constant interest rates, yield rates, and volatilities.}
\item {There are finitely many observation points, and two exercise levels (possibly $\infty$) at each observation point. If $S$ is above the upper exercise level or below the lower exercise level at any observation point, the option is exercised and the payoff is a linear function in $S$.}
\item {At maturity, if $S$ is between the two exercise levels, a payoff is incurred which is also a linear function in $S$.}
\end{enumerate}

\section{Outline of the method}
\label{sec_outl}

Let $V(t,S)$ denote the value of the option (as a function of asset price $S$) at any time $t$, and let
\[
  V_m(S) = V(t_m,S),\qquad m = 0,1,\dotsc,M,
\]
denote the value of the option at the observation dates. Our goal is to find $V_0(S(t_0))$, and our strategy is to use backward induction in time. Since $V_M(S)$ is piecewise linear in $S$, $V_{M-1}(S)$ has a simple explicit expression. For each $m = M-1,\dotsc,1$, we write $V_{m-1}(S)$ as the risk-neutral expectation of $V_m(S)$ for $K_{m-1}^- < S < K_{m-1}^+$, and as $a_{m-1}^\pm S + b_{m-1}^\pm$ otherwise. The expectation is given by an explicit integral and is calculated numerically. The core of the quadrature method is the calculation of $M-1$ expectation integrals step-by-step. Let
\[
  \tau_m = \frac{1}{2}\, \sigma_m^2 \Delta t_m = \frac{1}{2}\, \sigma_m^2 (t_m - t_{m-1}).
\]
For each $1 \leq m \leq M-1$, note that $S(t_m)$ has a lognormal distribution as in \eqref{stdis}. The relevant probability density functions are known to be \cite{log}
\begin{equation}
  \label{defrho}
  \rho_m(y,S) = \frac{1}{2\sqrt{\pi\tau_m}\, y} \exp \biggl\{ -\frac{1}{4\tau_m} \Bigl( \log \frac{y}{S} - \frac{2}{\sigma_m^2} \bigl[ r_m-q_m - \tfrac{1}{2} \sigma_m^2 \bigr] \tau_m \Bigr)^2 \biggr\}.
\end{equation}
For simplicity of notations we define $K_0^+ = \infty$ and $K_0^- = 0$. By the fundamental theorems of asset pricing, we have the risk-neutral pricing formula \cite{sto}:
\begin{align}
  \label{exp}
  & V_{m-1}(S) = e^{-2r_m \tau_m/\sigma_m^2} \mathbb{E} \bigl[ V_m(\cdot)|S \bigr]
    = e^{-2r_m \tau_m/\sigma_m^2} \int_0^\infty V_m(y) \rho_m(y,S)\,dy \\
  & = \frac{e^{-2r_m \tau_m/\sigma_m^2}}{2\sqrt{\pi\tau_m}} \int_0^\infty \frac{1}{y}\, V_m(y) \exp \biggl\{ -\frac{1}{4\tau_m} \Bigl( \log \frac{y}{S} - \frac{2}{\sigma_m^2} \bigl[ r_m-q_m - \tfrac{1}{2} \sigma_m^2 \bigr] \tau_m \Bigr)^2 \biggr\}\,dy, \nonumber
\end{align}
for $K_{m-1}^- < S < K_{m-1}^+$ and $1 \leq m \leq M-1$. By Assumption (B) from Section \ref{sec_assump}, we also have
\begin{equation}
  \label{correct}
  V_{m-1}(S) =
  \begin{cases}
    a_{m-1}^- S + b_{m-1}^-, & S \leq K_{m-1}^- \\
    a_{m-1}^+ S + b_{m-1}^+, & S \geq K_{m-1}^+
  \end{cases}.
\end{equation}
To further study the formulas \eqref{exp} and \eqref{correct}, we first recall some classical results on the pricing of binary options.

\begin{lem}
\label{bin}
Let $K > 0$, and let $\chi_A$ denote the characteristic function of a set $A$. Consider an option with no early-exercise features.
\begin{enumerate}
\item {If the option has payoff $\hat{V}_m(y) = \chi_{[K,\infty)} y$, then $\hat{V}_{m-1}(S) = V_m^a(S,K,1)$.}
\item {If $\hat{V}_m(y) = \chi_{(0,K]} y$, then $\hat{V}_{m-1}(S) = V_m^a(S,K,-1)$.}
\item {If $\hat{V}_m(y) = \chi_{[K,\infty)}$, then $\hat{V}_{m-1}(S) = V_m^b(S,K,1)$.}
\item {If $\hat{V}_m(y) = \chi_{(0,K]}$, then $\hat{V}_{m-1}(S) = V_m^b(S,K,-1)$.}
\end{enumerate}
The functions $V_m^a$ and $V_m^b$ are defined as
\[
  V_m^a(S,K,\epsilon) = e^{-2q_m \tau_m/\sigma_m^2} S N(\epsilon d_1),\qquad V_m^b(S,K,\epsilon) = e^{-2r_m \tau_m/\sigma_m^2} N(\epsilon d_2),
\]
where $N$ is the cumulative normal distribution function, and
\[
  d_1 = \frac{1}{\sqrt{2\tau_m}} \Bigl( \log \frac{S}{K} + \frac{2}{\sigma_m^2} \bigl[ r_m-q_m + \tfrac{1}{2} \sigma_m^2 \bigr] \tau_m \Bigr),\qquad d_2 = d_1 - \sqrt{2\tau_m}.
\]
\end{lem}
\begin{proof}
By definition $V_m^a$ is the value of an asset-or-nothing option, and $V_m^b$ is the value of a cash-or-nothing option. The valuation formulas are just standard results for binary options \cite{book}.
\end{proof}

\begin{remark}
The standard Black-Scholes formulas in Lemma \ref{bin} ignore the possible effects of volatility smiles. If such effects need to be taken into account, one may amend the definitions of $V_m^a$ and $V_m^b$ (as given in the Lemma) by incorporating suitable vega-induced correction terms \cite{vols}. For instance, the value of a cash-or-nothing call option in the presence of volatility smiles would become
\[
  V_{\mathrm{smile}} = V_{\mathrm{no~smile}} - \frac{\partial V_{\mathrm{vanilla}}}{\partial \sigma} \frac{\partial \sigma}{\partial K}.
\]
\end{remark}

We can use Lemma \ref{bin} to obtain an explicit formula for the value of $V_{M-1}(S)$. By Assumption (B)--(C) from Section \ref{sec_assump}, we have
\begin{equation}
  \label{vm}
  V_M(S) =
  \begin{cases}
    a_M^- S + b_M^-, & S \leq K_M^- \\
    a_M S + b_M, & K_M^- < S < K_M^+ \\
    a_M^+ S + b_M^+, & S \geq K_M^+
  \end{cases}.
\end{equation}
Without loss of generality we may assume $0 < K_M^\pm < \infty$, since otherwise we may choose some arbitrary $0 < K_M^\pm < \infty$ and set $a_M^\pm = a_M,\ b_M^\pm = b_M$.

\begin{prop}
The value of the option at $t_{M-1}$ is given by
\label{vm1}
\begin{align*}
  \tilde{V}_{M-1}(S) & = a_M^- V_M^a(S,K_M^-,-1) + b_M^- V_M^b(S,K_M^-,-1) \\
  &\qquad{} + a_M \bigl[ V_M^a(S,K_M^-,1) - V_M^a(S,K_M^+,1) \bigr] \\
  &\qquad{} + b_M \bigl[ V_m^b(S,K_M^-,1) - V_M^b(S,K_M^+,1) \bigr] \\
  &\qquad{} + a_M^+ V_M^a(S,K_M^+,1) + b_M^+ V_M^b(S,K_M^+,1),
\end{align*}
for $K_{M-1}^- < S < K_{M-1}^+$.
\end{prop}
\begin{proof}
Clearly, the option from $t_{M-1}$ to $t_M$ is equivalent to a linear combination of binary options consisting of two put options with strike $K_M^-$, two call options with strike $K_M^-$ and four call options with strike $K_M^+$. The result then follows from Lemma \ref{bin}.
\end{proof}

With the aid of \eqref{exp} and Proposition \ref{vm1}, we may write the main recursion of our quadrature method as follows.

\begin{prop}
\label{main}
Let $\tilde{V}_M = V_M$ be defined in \eqref{vm}, and $\tilde{V}_{m-1}$ be given by the following recursion formula:
\begin{align}
  \label{recur}
  \tilde{V}_{m-1}(S) & = e^{-2r_m \tau_m/\sigma_m^2} \int_{K_m^-}^{K_m^+} \tilde{V}_m(y) \rho_m(y,S)\,dy + a_m^+ V_m^a(S,K_m^+,1) \\
  &\qquad{} + b_m^+ V_m^b(S,K_m^+,1) + a_m^- V_m^a(S,K_m^-,-1) + b_m^- V_m^b(S,K_m^-,-1), \nonumber
\end{align}
for $1 \leq m \leq M$. Then we have $\tilde{V}_m(S) = V_m(S)$ for $K_m^- < S < K_m^+$ and $0 \leq m \leq M$. In particular, $\tilde{V}_0(S(t_0)) = V_0(S(t_0))$.
\end{prop}
\begin{proof}
This is merely the classical recursion formula for quadrature methods specialized to the Black-Scholes model. To prove the formula, we only need to show
\begin{equation}
  \label{vtil}
  V_m(S) = \tilde{V}_m(S) \chi_{(K_m^-,K_m^+)} + (a_m^+ S + b_m^+) \chi_{[K_m^+,\infty)} + (a_m^- S + b_m^-) \chi_{(0,K_m^-]},
\end{equation}
for all $0 \leq m \leq M$. By assumption \eqref{vtil} is true for $m = M$. Assume now \eqref{vtil} holds for some $1 \leq m \leq M$. Substituting the equation into \eqref{exp}, applying Lemma \ref{bin}, comparing the result with \eqref{recur} and using \eqref{correct}, we observe that \eqref{vtil} holds for $m-1$. The result then follows from induction.
\end{proof}

\begin{remark}
Recursion formula \eqref{recur} lies at the heart of our quadrature method and distinguishes our method from other quadrature methods, which are primarily based on \eqref{exp} or one of its many variants. The significance of the formula \eqref{recur} lies in the fact that it makes explicit use of Black-Scholes formulas to separate the expectation integral $\mathbb{E} \bigl[ V_m(\cdot)|S \bigr]$ into a ``quadrature part''
\[
  F_{m-1}(S) = e^{-2r_m \tau_m/\sigma_m^2} \int_{K_m^-}^{K_m^+} \tilde{V}_m(y) \rho_m(y,S)\,dy,
\]
and an ``early-exercise part''
\begin{align*}
  E_{m-1}(S) & = a_m^+ V_m^a(S,K_m^+,1) + b_m^+ V_m^b(S,K_m^+,1) \\
  &\qquad{} + a_m^- V_m^a(S,K_m^-,-1) + b_m^- V_m^b(S,K_m^-,-1).
\end{align*}
Since the function $\tilde{V}_m(S)$ is smooth for $S \in (K_m^-,K_m^+)$ (in fact for all $S \in (0,\infty)$, as we will show below), the integral $F_{m-1}(S)$ can be evaluated accurately and efficiently using a high-order quadrature method such as Simpson's rule. In contrast, the integrand $V_m(S)$ in the original recursion formula \eqref{exp} is discontinuous on $(0,\infty)$ (in either $V_m$ itself or in its first derivative); this makes the accurate evaluation of the expectation integral a difficult and challenging task. Although \eqref{recur} applies specifically to the Black-Scholes model, the same idea can be used for other asset price models, as long as a suitable analytical formula (exact or approximate) can be found for the probability density function $\rho_m(y,S)$ and early-exercise part $E_{m-1}(S)$.
\end{remark}

With $\tilde{V}_{M-1}(S)$ given in Proposition \ref{vm1}, Proposition \ref{main} implies that we may apply \eqref{recur} successively to obtain $\tilde{V}_{M-2}(S),\tilde{V}_{M-3}(S),\dotsc,\tilde{V}_0(S)$. The value of the option is equal to $\tilde{V}_0(S(t_0))$.

\section{Details of Implementation}
\label{sec_impl}

In \eqref{recur}, $\tilde{V}_{m-1}(S)$ is written as a sum of explicit functions and an integral, namely
\begin{align}
  \label{recur1}
  \tilde{V}_{m-1}(S) & = a_m^+ V_m^a(S,K_m^+,1) + b_m^+ V_m^b(S,K_m^+,1) \\
  &\qquad{} + a_m^- V_m^a(S,K_m^-,-1) + b_m^- V_m^b(S,K_m^-,-1) + F_{m-1}(S), \nonumber
\end{align}
where
\[
  F_{m-1}(S) = e^{-2r_m \tau_m/\sigma_m^2} \int_{K_m^-}^{K_m^+} \tilde{V}_m(y) \rho_m (y,S)\,dy.
\]
We may truncate the integral by replacing its upper and lower bounds by
\[
  L_m^+ = \min \bigl\{ K_m^+, S(t_0) C \bigr\},\qquad \text{and}\qquad L_m^- = \max \bigl\{ K_m^-, S(t_0)/C \bigr\},
\]
respectively, where $C > 1$ is a suitable constant. In practice, the choice
\[
  \log C = 10\sigma_0 \sqrt{t_M-t_0} + \bigl( 1 + \tfrac{1}{2} \sigma_0^2 \bigr) (t_M-t_0),
\]
where $\sigma_0 = \max_{1 \leq m \leq M} \sigma_m$, is sufficient to reduce the truncation errors to round-off level. Heuristically this is clear from \eqref{stdis}, which suggests that the chance that $S(t_m)$ move outside the range $(S(t_0)/C,S(t_0) C)$ is negligibly small. The rigorous derivation of the error bounds can be obtained using a recursive argument, as will be explained in the Appendix.

Now we consider the truncated integral
\[
  \tilde{F}_{m-1}(S) = e^{-2r_m \tau_m/\sigma_m^2} \int_{L_m^-}^{L_m^+} \tilde{V}_m(y) \rho_m(y,S)\,dy.
\]
If $K_m^- \geq S(t_0) C$ or $K_m^+ \leq S(t_0)/C$, then by convention the integral is zero. Thus in what follows we shall assume $K_m^- < S(t_0)C$ and $K_m^+ > S(t_0)/C$. Let
\[
  B_m^\pm = \log \frac{L_m^\pm}{S(t_0)},
\]
and denote
\begin{align*}
  S = S(t_0) e^x,&\qquad y = S(t_0) e^z, \\
  \alpha_m = \frac{1}{\sigma_m^2} \bigl[ r_m-q_m - \tfrac{1}{2} \sigma_m^2 \bigr],&\qquad \beta_m = \frac{1}{\sigma_m^4} \bigl[ r_m-q_m - \tfrac{1}{2} \sigma_m^2 \bigr]^2 + \frac{2r_m}{\sigma_m^2}, \\
  u_m(x) = \tilde{V}_m(S(t_0) e^x),&\qquad w_m(x) = \exp \Bigl\{ -\frac{x^2}{4\tau_m} - \alpha_m x \Bigr\}.
\end{align*}
The truncated integral can be rewritten as
\begin{equation}
  \label{int}
  \tilde{F}_{m-1}(S(t_0) e^x) = \frac{e^{-\beta_m \tau_m}}{2\sqrt{\pi\tau_m}} \int_{B_m^-}^{B_m^+} w_m(x-z) u_m(z)\,dz.
\end{equation}
One can show by differentiating \eqref{int} that $\tilde{F}_m$, and thus $\tilde{V}_m$ and $u_m$, are smooth functions in $x$. This means we can compute the integrals efficiently using a high-order quadrature such as Simpson's rule.

In general, $B_m^\pm$ are different for different values of $m$, so they cannot all be placed on one grid. Now we choose a uniform grid $\mathbf{x} = \{x_1,x_2,\dotsc,x_{N}\}$, where $x_1 = -\log C$ and $x_N = \log C$. Let
\[
  h = \frac{x_N - x_1}{N-1} = \frac{2\log C}{N-1}.
\]
For each $m$, let
\[
  p_m^- = \min \bigl\{ i\colon x_i \geq B_m^- \bigr\},\qquad p_m^+ = \max \bigl\{ i\colon x_i < B_m^+ \bigr\},
\]
where by definition $p_m^- \geq 1$ and $p_m^+ < N$. Since we will use Simpson's rule which requires an odd number of grid points, we define
\[
  p_0 = (p_m^+ - p_m^-) \mod 2,
\]
and rewrite \eqref{int} as
\begin{align}
  \label{int1}
  \tilde{F}_{m-1}(S(t_0) e^x) & = \frac{e^{-\beta_m \tau_m}}{2\sqrt{\pi\tau_m}} \biggl( \int_{x_{p_m^-}}^{x_{p_m^+ + p_0}} w_m(x-z) u_m(z)\,dz \\
  &\qquad{} + \int_{B_m^-}^{x_{p_m^-}} w_m(x-z) u_m(z)\,dz + \int_{x_{p_m^+ + p_0}}^{B_m^+} w_m(x-z) u_m(z)\,dz \biggr). \nonumber
\end{align}
For each $2 \leq m \leq M-1$, we will compute $\tilde{F}_{m-1}(S(t_0) e^x)$ for all
\[
  x \in \bigl\{ x_1,x_2,\dotsc,x_{N}, B_{m-1}^-, B_{m-1}^+, \xi_{m-1}^-, \xi_{m-1}^+ \bigr\},
\]
where
\[
  \xi_{m-1}^- = \frac{1}{2}\, (x_{p_{m-1}^-} + B_{m-1}^-),\qquad \xi_{m-1}^+ = \frac{1}{2}\, (x_{p_{m-1}^+ + p_0} + B_{m-1}^+).
\]
For $m = 1$ we only need to compute $\tilde{F}_{m-1}(S(t_0) e^x)$ for $x = 0$, since the value of the option is given by $\tilde{V}_0(S(t_0))$.

\subsection{Computation of the first integral in \eqref{int1}}

To compute the first integral in \eqref{int1} using Simpson's rule, we let
\[
  U_m(i) =
  \begin{cases}
    u_m(x_i), & i = p_m^-, p_m^+ + p_0 \\
    4 u_m(x_i), & i = p_m^- + 1, p_m^- + 3, \dotsc, p_m^+ + p_0 - 1 \\
    2 u_m(x_i), & i = p_m^- + 2, p_m^- + 4, \dotsc, p_m^+ + p_0 - 2
  \end{cases}.
\]
The integral is discretized as
\begin{equation}
  \label{dis1}
  \int_{x_{p_m^-}}^{x_{p_m^+ + p_0}} w_m(x-z) u_m(z)\,dz = \frac{h}{3} \sum_{i=p_m^-}^{p_m^+ + p_0} w_m(x-x_i) U_m(i) + O(h^4),
\end{equation}
since Simpson's rule is of order 4 \cite{num}. Note that $U_m(i)$ is known from the previous step (or by Proposition \ref{vm1} for $m = M-1$) for all $i = 1,2,\dotsc,N$. For $x \in \{B_{m-1}^\pm, \xi_{m-1}^\pm, 0\}$, the sum \eqref{dis1} can be computed directly with complexity $O(N)$. For all grid points $x \in \{x_1,x_2,\dotsc,x_{N}\}$, on the other hand, the sum \eqref{dis1} can be computed altogether using FFT with complexity $O(N\log N)$. This latter fact is crucial to the efficient implementation of our quadrature method and is a consequence of the following simple observation.

\begin{prop}
\label{intfft}
Define $(2N-1)$-periodic grid functions $\hat{z},\ \hat{U}_m$, and $\hat{F}_m$ by
\begin{align*}
  \hat{z}(i) & = z_i = -2\log C + (i-1)h,\qquad 1 \leq i \leq 2N-1, \\
  \hat{U}_m(i) & =
  \begin{cases}
    0, & 1 \leq i < p_m^- \\
    U_m(i), & p_m^- \leq i \leq p_m^+ + p_0 \\
    0, & p_m^+ + p_0 < i \leq 2N-1
  \end{cases}, \\
  \hat{F}_m & = \mathcal{F}^{-1} \Bigl\{ \mathcal{F} \bigl( w_m(\hat{z}) \bigr) \mathcal{F}(\hat{U}_m) \Bigr\},
\end{align*}
where $\mathcal{F}$ and $\mathcal{F}^{-1}$ denote the discrete Fourier transform and the inverse discrete Fourier transform of size $2N-1$, respectively. Then
\[
  \int_{x_{p_m^-}}^{x_{p_m^+ + p_0}} w_m(x_j-z) u_m(z)\,dz = \frac{h}{3}\, \hat{F}_m(j+N) + O(h^4),
\]
for all $1 \leq j \leq N$. The above discrete Fourier transforms and inverse discrete Fourier transform can be calculated using FFT, and the total computational complexity is $O(N\log N)$.
\end{prop}
\begin{proof}
We consider the discrete convolution
\[
  G_m(j) = \sum_{i=1}^{2N-1} w_m(z_{j-i}) \hat{U}_m(i),
\]
for $j \in \mathbb{Z}$. Note that by definition,
\[
  z_{j+N-i} = -2\log C + (j-i+N-1)h = (j-i)h = x_j - x_i,
\]
for all $1-N \leq j-i \leq N-1$. Thus for $1 \leq j \leq N$, we have
\begin{align*}
  G_m(j+N) & = \sum_{i=1}^{2N-1} w_m(z_{j+N-i}) \hat{U}_m(i) \\
  & = \sum_{i=p_m^-}^{p_m^+ + p_0} w_m(z_{j+N-i}) U_m(i) = \sum_{i=p_m^-}^{p_m^+ + p_0} w_m(x_j - x_i) U_m(i).
\end{align*}
Therefore \eqref{dis1} with $x = x_j$ can be written as
\begin{equation}
  \label{dis1c}
  \int_{x_{p_m^-}}^{x_{p_m^+ + p_0}} w_m(x_j-z) u_m(z)\,dz = \frac{h}{3}\, G_m(j+N) + O(h^4).
\end{equation}
The discrete convolution $G_m$ can be calculated using FFT as
\[
  G_m = \mathcal{F}^{-1} \Bigl\{ \mathcal{F} \bigl( w_m(\hat{z}) \bigr) \mathcal{F}(\hat{U}_m) \Bigr\} = \hat{F}_m,
\]
with a complexity of $O(N\log N)$ \cite{num}.
\end{proof}

\begin{remark}
Our method differs from other well-known FFT-based methods (such as \cite{conv,fft}) in that we express the discrete quadrature rule \eqref{dis1} directly in terms of discrete Fourier transforms, instead of applying continuous Fourier transform to the integral and then discretizing the Fourier integrals (in other words, we have exchanged the order of Fourier transform and discretization). The direct application of the discrete Fourier transform (to the discrete quadrature rule) not only eliminates the need for artificially-introduced damping factors, which are required for the existence of the continuous Fourier transforms, but also eliminates the need for additional specially-designed computational grids which are required to satisfy Nyquist relations. This enables us to carry out the main recursion \eqref{recur} on a fixed uniform grid, without any additional artificial parameters.
\end{remark}

\subsection{Computation of the last two integrals in \eqref{int1}}

The last two integrals in \eqref{int1} are calculated in similar ways using Simpson's rule. First, note that we may use Proposition \ref{vm1} to calculate $\tilde{V}_{M-1}(S(t_0) e^x)$ for $x \in \{B_{M-1}^\pm, \xi_{M-1}^\pm\}$. Generally all four points are needed if the option has two barriers, and only two are needed if the option has one barrier. For each $2 \leq m \leq M-1$, assume $u_m(x) = \tilde{V}_m(S(t_0) e^x)$ has been calculated for $x \in \{B_m^\pm, \xi_m^\pm\}$. The last two integrals in \eqref{int1} are calculated using Simpson's rule as follows:
\begin{align}
  \label{intpart2}
  & \int_{B_m^-}^{x_{p_m^-}} w_m(x-z) u_m(z)\,dz = \frac{1}{6}\, (x_{p_m^-} - B_m^-) \bigl[ w_m(x-B_m^-) u_m(B_m^-) \\
  &\qquad\qquad\qquad{} + 4 w_m(x-\xi_m^-) u_m(\xi_m^-) + w_m(x-x_{p_m^-}) u_m(x_{p_m^-}) \bigr] + O(h^4), \nonumber \\
  \label{intpart3}
  & \int_{x_{p_m^+ + p_0}}^{B_m^+} w_m(x-z) u_m(z)\,dz = \frac{1}{6}\, (B_m^+ - x_{p_m^+ + p_0}) \bigl[ w_m(x-B_m^+) u_m(B_m^+) \\
  &\qquad\qquad\qquad{} + 4 w_m(x-\xi_m^+) u_m(\xi_m^+) + w_m(x-x_{p_m^+ + p_0}) u_m(x_{p_m^+ + p_0}) \bigr] + O(h^4). \nonumber
\end{align}

\section{Finding Optimal Exercise Prices for Bermudan options}
\label{sec_berms}

Unlike autocallable products and barrier options, Bermudan options do not have pre-specified exercise levels. Instead, one needs to solve for $K_m^\pm$ from the equations
\[
  \tilde{V}_m(K_m^+) = K_m^+ - K,
\]
for call options and
\[
  \tilde{V}_m(K_m^-) = K - K_m^-,
\]
for put options, where $\tilde{V}_m$ is determined by \eqref{recur}. For simplicity we assume the yield rates $q_m \geq 0$, which is almost always the case in practice. We will demonstrate how to find $K_m^-$, as the same procedure applies to $K_m^+$. Let
\[
  p = \min \bigl\{ i\colon \tilde{V}_m(S(t_0) e^{x_i}) > K - S(t_0) e^{x_i} \bigr\}.
\]
If $p = 1$ there is no early exercise, so $K_m^- = 0$. Otherwise we have $S(t_0) e^{x_{p-1}} \leq K_m^- < S(t_0) e^{x_{p}}$ by Corollary \ref{mono}. The value of $K_m^-$ can be found using classical root-finding methods such as the bisecting method or the secant method. Note that the bisecting method is guaranteed to converge by Corollary \ref{mono}, and it takes $O(\log N)$ steps to reduce the error of the approximate root to an order of $O(h^4)$. Since the cost for calculating $\tilde{V}_m$ at one point using \eqref{recur1} is $O(N)$, the total cost for finding the optimal exercise price is $O(N\log N)$. The secant method is superlinear and converges faster than the bisecting method, though its error estimates are not as straightforward.

\section{Summary of the Algorithm}
\label{sec_algo}

We summarize our algorithm as follows:

\begin{algorithmic}[1]
\STATE Define the functions $V_m^a, V_m^b$ as in Lemma \ref{bin}
\STATE Define the function $\tilde{V}_{M-1}$ as in Proposition \ref{vm1}
\IF {option style is Bermudan}
    \STATE {Calculate $K_{M-1}^\pm$ as in Section \ref{sec_berms}}
\ENDIF
\STATE Calculate $p_{M-1}^\pm$ and $p_0$ to find the bounds of integration in \eqref{int1}
\STATE Use Proposition \ref{vm1} to compute $\tilde{V}_{M-1}(S(t_0) e^x)$ for $x \in \{B_{M-1}^\pm, \xi_{M-1}^\pm\}$, and assign their values to $v_1^\pm, v_2^\pm$ respectively
\STATE Define a vector $\mathbf{S}$ as $S(i) \gets S(t_0) e^{x_i}$ for $i = 1,2,\dotsc,N$
\STATE Define a vector $\mathbf{y}$ as $y(i) \gets \tilde{V}_{M-1}(S(i))$ for $i = 1,2,\dotsc,N$
\FOR {$m = M-1$ downto 2}
    \STATE{Let $\hat{z}$ and $\hat{U}_m$ be as defined in Proposition \ref{intfft}}
    \STATE{$\hat{F}_m \gets \mathcal{F}^{-1} \Bigl\{ \mathcal{F} \bigl( w_m(\hat{z}) \bigr) \mathcal{F}(\hat{U}_m) \Bigr\}$}
    \STATE{Define (or redefine) the vector $\mathbf{Y_1}$ as $Y_1(j) \gets \frac{h}{3} \hat{F}_m(j+N)$ for $j = 1,2,\dotsc,N$}
    \STATE{Use \eqref{intpart2}, \eqref{intpart3}, and the values of $v_1^\pm, v_2^\pm, y(p_m^-), y(p_m^+ + p_0)$ to compute the last two integrals in \eqref{int1} at $x_i$ for $i = 1,2,\dotsc,N$, and assign their values to $\mathbf{Y_2}$ and $\mathbf{Y_3}$}

    \IF {option style is Bermudan}
        \STATE {Calculate $K_{m-1}^\pm$ as in Section} \ref{sec_berms}
    \ENDIF
    \STATE Calculate $p_{m-1}^\pm$ and $p_0$ to find the bounds of integration in \eqref{int1}
    \STATE Compute $\tilde{V}_{m-1}(S(t_0) e^x)$ for $x \in \{B_{m-1}^\pm, \xi_{m-1}^\pm\}$ using \eqref{recur1} (where the first integral in \eqref{int1} is computed using \eqref{dis1}, and the last two integrals using \eqref{intpart2}, \eqref{intpart3}, and the existing values of $v_1^\pm, v_2^\pm, y(p_m^-), y(p_m^+ + p_0)$), and assign their values to $v_1^\pm, v_2^\pm$ respectively

    \STATE{
    \begin{align*}
      \mathbf{y} & \gets \mathbf{Y_1} + \mathbf{Y_2} + \mathbf{Y_3} + a_m^+ V_m^a(\mathbf{S},K_m^+,1) + b_m^+ V_m^b(\mathbf{S},K_m^+,1) \\
      &\qquad\qquad{} + a_m^- V_m^a(\mathbf{S},K_m^-,-1) + b_m^- V_m^b(\mathbf{S},K_m^-,-1),
    \end{align*}
    as in \eqref{recur1} (note that $\mathbf{y}$ now stores $\tilde{V}_{m-1}(S(i))$ for $i = 1,2,\dotsc,N$)}
\ENDFOR
\STATE Compute $\tilde{V}_0(S(t_0))$ using \eqref{recur1}, where the first integral in \eqref{int1} is computed using \eqref{dis1}, and the last two integrals using \eqref{intpart2}, \eqref{intpart3}, and the existing values of $v_1^\pm, v_2^\pm, y(p_1^-), y(p_1^+ + p_0)$
\end{algorithmic}

Since the computational complexity of each step of the loop is $O(N\log N)$, the total complexity is $O(MN\log N)$.

\begin{remark}
While other quadrature methods typically employ multiple uniform grids or specially-designed (nonuniform or shifted) grids, our method utilizes only a fixed one-dimensional uniform grid, which not only eliminates the need for complicated inter-grid data transfer procedures, but also eliminates the need for special subroutines that are often required to interpolate data across discontinuities. This makes our method particularly easy to implement.
\end{remark}

\section{Numerical Examples}
\label{sec_exmp}

We will demonstrate the accuracy and efficiency of the proposed method using two examples, in which the value of an autocallable structured product and that of a double barrier option with time-dependent barriers are found.

\subsection{Example 1: Autocallable Structured Product}

We consider a knock-out autocallable structured product maturing in one year. The price of the underlying asset is 3000, the nominal amount is 1, and the volatility is 20\%. The observation dates (in years from now), barrier levels, and risk-free rates (in \%) are given below in Table \ref{tab_auto}.
\begin{table}[ht]
  \centering
  \caption{An autocallable structured product.}
  \begin{tabular}{ccc}
    \toprule
    Observation date & Barrier level & Risk-free rate \\
    \midrule
    0.2 & 3050 & 2 \\
    0.4 & 3100 & 2.1 \\
    0.6 & 3150 & 2.2 \\
    0.8 & 3200 & 2.3 \\
    1 & 3250 & 2.4 \\
    \bottomrule
  \end{tabular}
  \label{tab_auto}
\end{table}%

If the asset price reaches or goes above the barrier level at some observation date $t$, the investor receives a payment of $4\% \times t$. If the asset price is below the barrier at every observation date, the investor will have to pay a premium of $1\%$. The relative errors of the computed option values with varying grid sizes are shown below in Figure \ref{fig_auto_err_q}, where the exact option value is taken to be the one computed on the grid of size 70,001.
\begin{figure}[ht]
  \centering
  \subfigure[]{
    \includegraphics[scale=0.45]{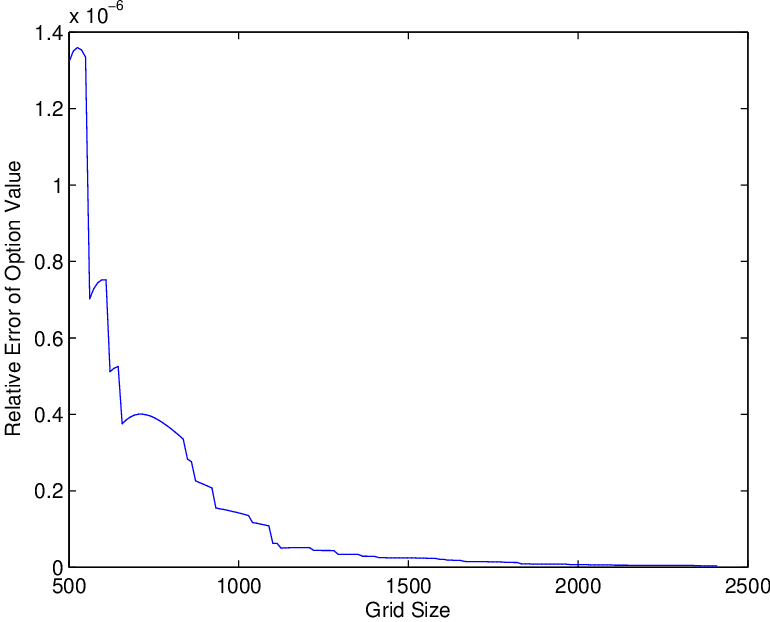}
    \label{fig_auto_err_q}
  }\hspace{2mm}
  \subfigure[]{
    \includegraphics[scale=0.45]{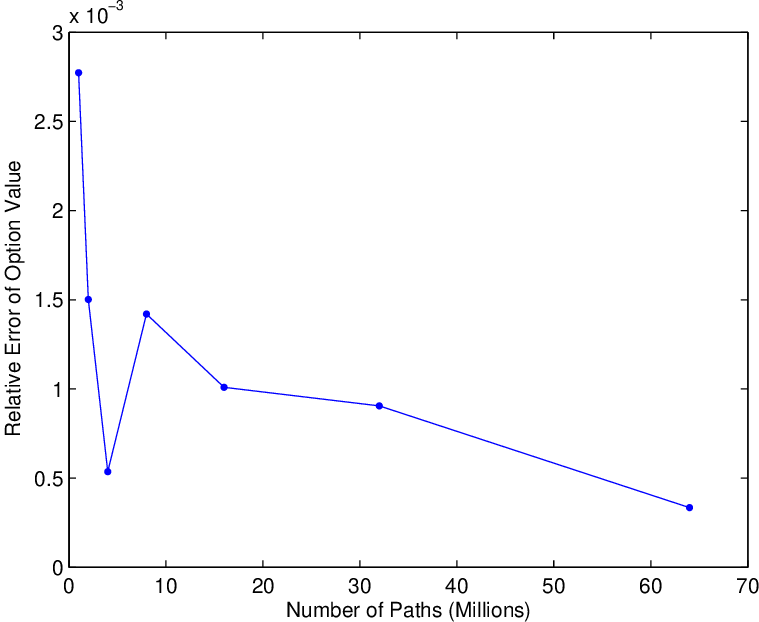}
    \label{fig_auto_err_mc}
  }
  \caption{Errors for the autocallable structured product, computed using (a) the proposed method and (b) Monte-Carlo simulations.}
\end{figure}%

As a comparison, the relative errors of the option values computed using Monte-Carlo simulations with antithetic variates technique are shown in Figure \ref{fig_auto_err_mc}. As is clear from the figures, the error of the option value computed using the proposed method is well within $10^{-5}$ with just 501 points in the grid, and drops very quickly as the grid size increases. In contrast, it takes more than ten million paths for Monte-Carlo simulations to reduce the error of the computed option value to within $10^{-3}$, and the error decays very slowly as the number of paths increases.

Figure \ref{fig_auto_time} below shows the CPU time used by the proposed method to price the autocallable structured product, where the code is developed in Matlab and is run on a personal computer.
\begin{figure}[ht]
  \centering
  % \includegraphics[width=10cm, height=8cm]{time1.eps}
  \includegraphics[scale=0.45]{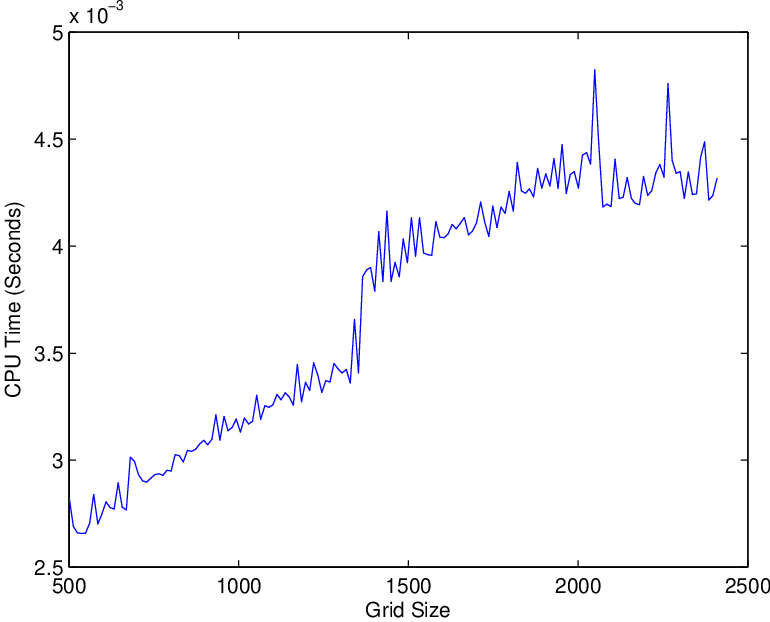}
  \caption{CPU time for autocallable structured product valuation, using the proposed method.}
  \label{fig_auto_time}
\end{figure}%
As is clear from the figure, the CPU time required by the proposed method is well within 0.01 seconds, and it increases approximately linearly as grid size increases. It is difficult to compare the speed of the proposed method with that of Monte-Carlo simulations directly, since the CPU time required by the latter depends largely on specific implementations. Nevertheless, the typical CPU time consumed by a Monte-Carlo simulation with tens of millions of paths ranges from tens of seconds to a few minutes.

\subsection{Example 2: Double Barrier Option}

As another example, consider a knock-out double barrier put option with time-dependent barrier levels. The price of the underlying asset is 2500, the strike price is 2600, the nominal amount is 1, and the volatility is 25\%. The option matures in two years. The observation dates (in years from now), barrier levels, and risk-free rates (in \%) are given below in Table \ref{tab_barr}.
\begin{table}[ht]
  \centering
  \caption{A double barrier option.}
  \begin{tabular}{cccc}
    \toprule
    Observation date & Barrier level 1 & Barrier level 2 & Risk-free rate \\
    \midrule
    0.25 & 2200 & 2800 & 1 \\
    0.50 & 2100 & 2900 & 1.1 \\
    0.75 & 2000 & 3000 & 1.2 \\
    1 & 1900 & 3100 & 1.3 \\
    1.25 & 1800 & 3200 & 1.2 \\
    1.50 & 1700 & 3300 & 1.3 \\
    1.75 & 1600 & 3400 & 1.4 \\
    2 & $-$ & $-$ & 1.5 \\
    \bottomrule
  \end{tabular}
  \label{tab_barr}
\end{table}%

If the asset price falls below barrier level 1 or rises above barrier level 2 at any observation date, the option ceases to exist. If the option is still valid at maturity, the payoff is the same as that of a vanilla put option. The relative errors of the computed option values with varying grid sizes are shown below in Figure \ref{fig_barr_err_q}, where the exact option value is taken to be the one computed on the grid of size 50,001.
\begin{figure}[ht]
  \centering
  \subfigure[]{
    \includegraphics[scale=0.45]{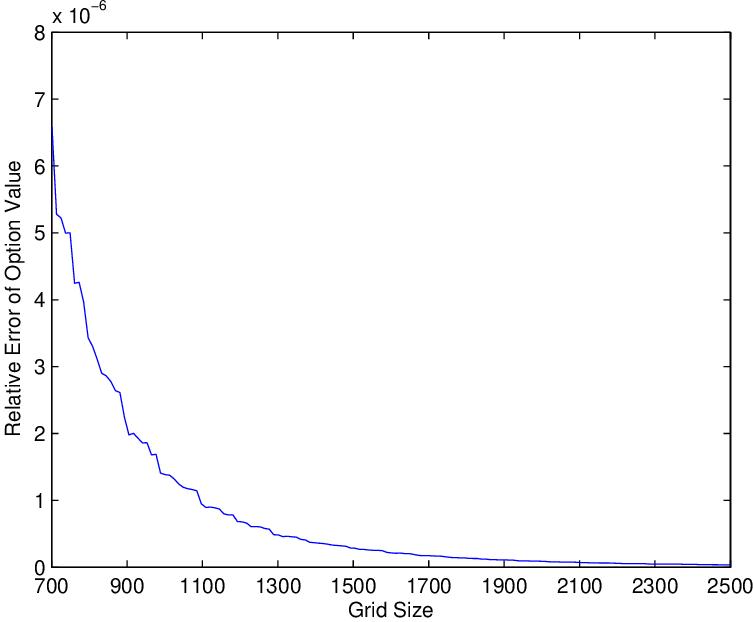}
    \label{fig_barr_err_q}
  }\hspace{2mm}
  \subfigure[]{
    \includegraphics[scale=0.45]{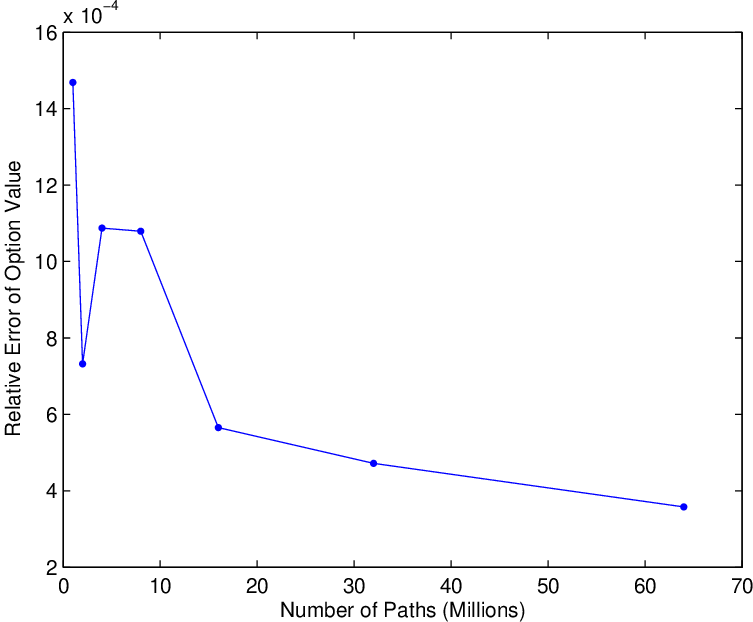}
    \label{fig_barr_err_mc}
  }
  \caption{Errors for the double barrier option, computed using (a) the proposed method and (b) Monte-Carlo simulations.}
\end{figure}%

As a comparison, the relative errors of the option values computed using Monte-Carlo simulations with antithetic variates technique are shown in Figure \ref{fig_barr_err_mc}. As is clear from the figures, the error of the option value computed using the proposed method is within $10^{-5}$ with just 701 points in the grid, and drops very quickly as the grid size increases. In contrast, it takes more than ten million paths for Monte-Carlo simulations to reduce the error of the computed option value to within $10^{-3}$, and the error decays very slowly as the number of paths increases.

Figure \ref{fig_barr_time} below shows the CPU time used by the proposed method to price the double barrier option, where the code is developed in Matlab and is run on a personal computer.
\begin{figure}[ht]
  \centering
  % \includegraphics[width=9.6cm, height=7.7cm]{time2.eps}
  \includegraphics[scale=0.45]{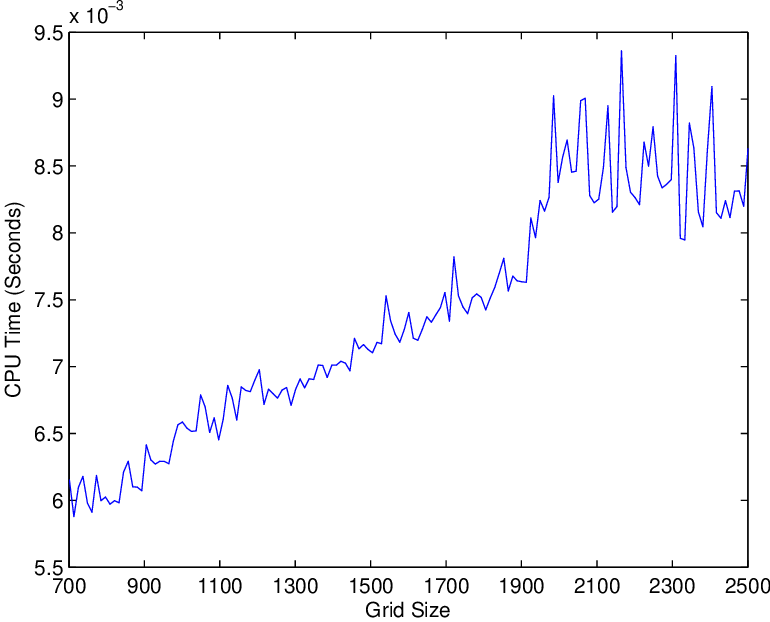}
  \caption{CPU time for double barrier option valuation, using the proposed method.}
  \label{fig_barr_time}
\end{figure}%
As is clear from the figure, the CPU time required by the proposed method is within 0.01 seconds, and it increases approximately linearly as grid size increases. In contrast, the typical CPU time consumed by a Monte-Carlo simulation with tens of millions of paths ranges from tens of seconds to a few minutes.

\begin{remark}
Although the designed order of the proposed method is 4, in the above numerical examples, a lower order of convergence (close to 3) is actually observed for the grid sizes considered. It is interesting to note that this apparent ``loss'' of order of accuracy is not a defect of our method; rather, it is a manifestation of the subtle influences that barrier levels can have on option pricing algorithms. These influences can be understood from two perspectives. First, in the above examples, the barrier levels $K_m^\pm$ are close to the spot price $S(t_0)$. This gives rise to a relatively small integration domain $[B_m^-,B_m^+]$ compared with the entire computational domain $[{-}\log C,\log C]$ (recall that
\[
  B_m^+ - B_m^- = \log \frac{L_m^+}{S(t_0)} - \log \frac{L_m^-}{S(t_0)} = \log \frac{L_m^+}{L_m^-} \leq \min \Bigl\{ \log \frac{K_m^+}{K_m^-}, \log \frac{K_m^+}{C^{-1}} \Bigr\}
\]
), which means that the set of grid points that are available for the discrete quadrature rule \eqref{dis1} represents only a relatively small fraction of the set of grid points introduced on the entire computational domain. Secondly, in the above examples the option prices $V_{m+1}(S)$ contain discontinuities at each observation date $t_{m+1}$. These discontinuities necessarily show up in the form of large gradients in the (smooth) functions $\tilde{V}_m(S)$ (via the expectation integrals $\mathbb{E} \bigl[ V_{m+1}(\cdot)|S \bigr]$), which means that the discrete quadrature rule \eqref{dis1} is being applied to fast-varying functions with only a relatively small number of grid points, leaving the integrands only marginally resolved and hence explaining the degeneracy observed in the convergence rate. If the barrier levels $K_m^\pm$ are pushed farther away from the spot price $S(t_0)$, so that the option becomes increasingly like a vanilla option, then the discrete quadrature rule \eqref{dis1} effectively applies to slow-varying functions with a relatively large number of grid points, which improves the resolution of the integrands and hence the convergence rate (to close to 4). Despite these caveats on convergence rate, we emphasize that our method is capable of pricing a sophisticated discretely monitored option and obtaining five to six significant digits within a fraction of a second, while at the same time being very easy to understand and to implement. Thus the (relatively technical) issue of convergence rate should pose no real concerns in practice.
\end{remark}

\begin{remark}
Although a relative error of the order $10^{-5}$ or $10^{-6}$ may not always seem necessary for option pricing problems considered in the real financial world, this extra accuracy is actually needed in the calculation of the Greeks, which are typically approximated by finite difference formulas and which are much more sensitive to numerical errors incurred in the calculation of option prices.
\end{remark}

\section{Appendix}

\subsection{Estimate of Truncation Errors}

To estimate the truncation error for the integral in \eqref{recur1}, we first introduce

\begin{lem}
\label{bound}
Let
\begin{align*}
  A = \max \bigl\{ |a_1^\pm|, \dotsc, |a_M^\pm|, |a_M| \bigr\},&\qquad B = \max \bigl\{ |b_1^\pm|, \dotsc, |b_M^\pm|, |b_M| \bigr\}, \\
  R = \min \bigl\{ r_1, \dotsc, r_M, 0 \bigr\},&\qquad Q = \min \bigl\{ q_1, \dotsc, q_M, 0 \bigr\}.
\end{align*}
Then
\begin{align*}
  |\tilde{V}_m(S)| \leq e^{Q(t_m-t_M)} AS + e^{R(t_m-t_M)} B,&\qquad \forall 0 \leq m \leq M,\ \forall S \in (K_m^-,K_m^+), \\
  \intertext{and}
  |V_m(S)| \leq e^{Q(t_m-t_M)} AS + e^{R(t_m-t_M)} B,&\qquad \forall 0 \leq m \leq M,\ \forall S \in (0,\infty).
\end{align*}
\end{lem}
\begin{proof}
Clearly, by assumption (cf. \eqref{vm}),
\[
  |\tilde{V}_M(S)| = |V_M(S)| \leq AS + B,\qquad \forall S \in (0,\infty).
\]
Now assume
\begin{align*}
  |\tilde{V}_m(S)| \leq e^{Q(t_m-t_M)} AS + e^{R(t_m-t_M)} B,&\qquad \forall S \in (K_m^-,K_m^+), \\
  \intertext{and}
  |V_m(S)| \leq e^{Q(t_m-t_M)} AS + e^{R(t_m-t_M)} B,&\qquad \forall S \in (0,\infty),
\end{align*}
for some $1 \leq m \leq M$. By \eqref{exp} and Proposition \ref{main}, we have
\[
  \tilde{V}_{m-1}(S) = e^{-2r_m \tau_m/\sigma_m^2} \int_0^\infty V_m(y) \rho_m(y,S)\,dy,
\]
for $K_{m-1}^- < S < K_{m-1}^+$. By definition \eqref{defrho}, it is clear that $\rho_m(y,S) = \rho_m(y/S,1)/S$. Thus a simple change of variable $z = y/S$ yields
\begin{align*}
  & |\tilde{V}_{m-1}(S)| = e^{-2r_m \tau_m/\sigma_m^2} \Bigl| \int_0^\infty V_m(Sz) \rho_m(z,1)\,dz \Bigr| \\
  &\qquad{} \leq e^{-2r_m \tau_m/\sigma_m^2} \int_0^\infty \bigl[ e^{Q(t_m-t_M)} AS z + e^{R(t_m-t_M)} B \bigr] \rho_m(z,1)\,dz \\
  &\qquad{} = e^{-2r_m \tau_m/\sigma_m^2} \biggl( e^{Q(t_m-t_M)} AS \int_0^\infty z\rho_m(z,1)\,dz + e^{R(t_m-t_M)} B \biggr).
\end{align*}
The integral $\int_0^\infty z\rho_m(z,1)\,dz$ is the expectation of the lognormal distribution, which is simply $\exp\{2(r_m-q_m) \tau_m/\sigma_m^2\}$ \cite{log}. Thus we have (observe $R \leq 0$ and $Q \leq 0$)
\begin{align*}
  & |\tilde{V}_{m-1}(S)| \leq e^{-2q_m \tau_m/\sigma_m^2 + Q(t_m-t_M)} AS + e^{-2r_m \tau_m/\sigma_m^2 + R(t_m-t_M)} B \\
  &\qquad{} \leq e^{Q(t_{m-1}-t_M)} AS + e^{R(t_{m-1}-t_M)} B,
\end{align*}
for all $S \in (K_{m-1}^-,K_{m-1}^+)$. By \eqref{correct} and Proposition \ref{main}, we then deduce
\[
  |V_{m-1}(S)| \leq e^{Q(t_{m-1}-t_M)} AS + e^{R(t_{m-1}-t_M)} B,
\]
for all $S \in (0,\infty)$. The result then follows from induction.
\end{proof}

The possibly infinite integral in \eqref{recur} is approximated by a finite integral. To be specific, let $C > 1,\ S_0 = S(t_0)$, and $\tilde{G}_{M-1}(S) = \tilde{V}_{M-1}(S)$. We consider $\tilde{G}_{m-1}$ ($1 \leq m \leq M-1$) defined recursively by
\begin{align*}
  \tilde{G}_{m-1}(S) & = e^{-2r_m \tau_m/\sigma_m^2} \mathbb{E} \bigl[ \tilde{G}_m(\cdot)|S \bigr] \\
  & = e^{-2r_m \tau_m/\sigma_m^2} \int_{L_m^-}^{L_m^+} \tilde{G}_m(y) \rho_m(y,S)\,dy + a_m^+ V_m^a(S,K_m^+,1) \\
  &\qquad{} + b_m^+ V_m^b(S,K_m^+,1) + a_m^- V_m^a(S,K_m^-,-1) + b_m^- V_m^b(S,K_m^-,-1).
\end{align*}
A direct calculation shows that the errors $\tilde{R}_m = \tilde{V}_m - \tilde{G}_m$ satisfy the recursion
\begin{align*}
  \tilde{R}_{m-1}(S) & = e^{-2r_m \tau_m/\sigma_m^2} \biggl( \int_{L_m^-}^{L_m^+} \tilde{R}_m(y) \rho_m(y,S)\,dy \\
  &\qquad{} + \int_{K_m^-}^{L_m^-} \tilde{V}_m(y) \rho_m(y,S)\,dy + \int_{L_m^+}^{K_m^+} \tilde{V}_m(y) \rho_m(y,S)\,dy \biggr).
\end{align*}
We use the operator notation
\[
  \mathcal{T}_m(f)(S) = e^{-2r_m \tau_m/\sigma_m^2} \int_0^\infty f(y) \rho_m(y,S)\,dy,
\]
to write the recursion of $\tilde{R}_m$ as
\[
  \tilde{R}_{m-1} = \mathcal{T}_m(\tilde{R}_m \chi_{(L_m^-,L_m^+)}) + \mathcal{T}_m(\tilde{V}_m \chi_{(K_m^-,L_m^-] \cup [L_m^+,K_m^+)}),
\]
for $1 \leq m \leq M-1$. Note also that $\tilde{R}_{M-1}(S) = 0$.

Now we consider $\tilde{Q}_m$ defined by the recursion
\begin{equation}
  \label{qm}
  \tilde{Q}_{m-1} = \mathcal{T}_m(\tilde{Q}_m) + \mathcal{T}_m \bigl( (e^{Q(t_0-t_M)} Ay + e^{R(t_0-t_M)} B) \chi_{(0,S_0/C] \cup [S_0 C,\infty)} \bigr),
\end{equation}
and $\tilde{Q}_{M-1} = 0$. It is easy to show using Lemma \ref{bound} and induction that $|\tilde{R}_m| \leq \tilde{Q}_m$ for all $0 \leq m \leq M-1$. We can also apply the recursion formula \eqref{qm} to get the expansion
\[
  \tilde{Q}_0 = \sum_{m=1}^{M-1} \mathcal{T}_1 \circ \dotsb \circ \mathcal{T}_m \bigl( (e^{Q(t_0-t_M)} Ay + e^{R(t_0-t_M)} B) \chi_{(0,S_0/C] \cup [S_0 C,\infty)} \bigr).
\]
Since $\mathcal{T}_m$ is the risk-neutral expectation operator discounted from $t_m$ to $t_{m-1}$, $\mathcal{T}_1 \circ \dotsb \circ \mathcal{T}_m(f)$ is simply the value of a European option whose payoff at $t_m$ is given by $f$. Thus $\tilde{Q}_0$ is equal to the value of a sum of $2(M-1)$ binary call options with strike $S_0 C$ and $2(M-1)$ binary put options with strike $S_0/C$. As a result,
\begin{align}
  \label{q00}
  \tilde{Q}_0(S_0) & = e^{Q(t_0-t_M)} A \sum_{m=1}^{M-1} \bigl[ C_a(S_0,S_0 C,t_m) + P_a(S_0,S_0/C,t_m) \bigr] \\
  &\qquad{} + e^{R(t_0-t_M)} B \sum_{m=1}^{M-1} \bigl[ C_b(S_0,S_0 C,t_m) + P_b(S_0,S_0/C,t_m) \bigr], \nonumber
\end{align}
where $C_a, P_a, C_b, P_b$ denote values of asset-or-nothing call, asset-or-nothing put, cash-or-nothing call, and cash-or-nothing put respectively. According to Lemma \ref{bin}, the values of these binary options are given by
\begin{align*}
  C_a(S_0,S_0 C,t_m) = S_0 e^{-\bar{q}_m (t_m-t_0)} N(d_3),&\quad P_a(S_0,S_0/C,t_m) = S_0 e^{-\bar{q}_m (t_m-t_0)} N(d_4), \\
  C_b(S_0,S_0 C,t_m) = e^{-\bar{r}_m (t_m-t_0)} N(d_5),&\quad P_b(S_0,S_0/C,t_m) = e^{-\bar{r}_m (t_m-t_0)} N(d_6),
\end{align*}
where
\begin{align*}
  d_3 & = \frac{1}{\bar{\sigma}_m \sqrt{t_m-t_0}} \bigl[ -\log C + (\bar{r}_m-\bar{q}_m + \tfrac{1}{2} \bar{\sigma}_m^2) (t_m-t_0) \bigr], \\
  d_4 & = \frac{1}{\bar{\sigma}_m \sqrt{t_m-t_0}} \bigl[ -\log C - (\bar{r}_m-\bar{q}_m + \tfrac{1}{2} \bar{\sigma}_m^2) (t_m-t_0) \bigr], \\
  d_5 & = \frac{1}{\bar{\sigma}_m \sqrt{t_m-t_0}} \bigl[ -\log C + (\bar{r}_m-\bar{q}_m - \tfrac{1}{2} \bar{\sigma}_m^2) (t_m-t_0) \bigr], \\
  d_6 & = \frac{1}{\bar{\sigma}_m \sqrt{t_m-t_0}} \bigl[ -\log C - (\bar{r}_m-\bar{q}_m - \tfrac{1}{2} \bar{\sigma}_m^2) (t_m-t_0) \bigr],
\end{align*}
and $\bar{r}_m, \bar{q}_m, \bar{\sigma}_m$ represent the time-weighted averages of $\{r_n, q_n, \sigma_n\}_{n=1}^{m}$ respectively. Generally, in practice, the absolute values of annual interest and yield rates will not exceed $50\%$, $t_M-t_0$ will not exceed $10$ years, and $A$ will not exceed $1$. For general autocallable structured products, $B$ will not exceed $t_M-t_0$, and for Bermudan options $B$ will not exceed $K$, which is not much larger than $S_0$. We may also assume $M \leq 120$, which corresponds to products that are not too frequently monitored, say monthly (for more frequently monitored products, such as daily monitored products, continuity correction methods \cite{cont1,cont2} are usually more appropriate). Let $\sigma_0 = \max_{1 \leq m \leq M} \sigma_m$. If we choose
\[
  \log C = 10\sigma_0 \sqrt{t_M-t_0} + \bigl( 1 + \tfrac{1}{2} \sigma_0^2 \bigr) (t_M-t_0),
\]
we can make sure that
\[
  d_{3,4,5,6} \leq -10,\qquad \text{and thus}\qquad N(d_{3,4,5,6}) < 10^{-23}.
\]
A crude estimate using \eqref{q00} then shows that the error bound $\tilde{Q}_0(S_0)$ does not exceed $10^{-15}(S_0+1)$. This means the relative truncation error is negligible for all practical purposes.

\subsection{Analysis of $K_m^\pm$ for Bermudan Options}

The proper application of Proposition \ref{main} requires the uniqueness of the exercise prices $K_m^\pm$, which we now establish for Bermudan options.

To begin with, observe that the risk-neutral pricing formulas \eqref{exp}--\eqref{correct} applied to Bermudan options can be written in an alternative form as
\begin{equation}
  \label{vtv0}
  V_{m-1}(S) = \max \bigl\{ \tilde{V}_{m-1}(S), \epsilon (S-K) \bigr\},\qquad \forall 1 \leq m \leq M,\ \forall S \in (0,\infty),
\end{equation}
where
\[
  \epsilon =
  \begin{cases}
    1, & \text{if the option is a Bermudan call} \\
    -1, & \text{if the option is a Bermudan put}
  \end{cases},
\]
and
\begin{align}
  \label{exp1}
  \tilde{V}_{m-1}(S) & = e^{-2r_m \tau_m/\sigma_m^2} \int_0^\infty V_m(y) \rho_m(y,S)\,dy \\
  & = e^{-2r_m \tau_m/\sigma_m^2} \int_0^\infty V_m(Sz) \rho_m(z,1)\,dz. \nonumber
\end{align}
Since, by definition,
\[
  V_M(S) = \max \bigl\{ 0, \epsilon (S-K) \bigr\},
\]
\eqref{vtv0} and \eqref{exp1} define $V_m$ and $\tilde{V}_m$ recursively for all $0 \leq m \leq M-1$. It is easy to see that $V_m(S) \geq \tilde{V}_m(S) > 0$ for all $0 \leq m \leq M-1$ and $S \in (0,\infty)$.

\begin{prop}
\label{bermb}
Assume $q_m \geq 0$ for all $1 \leq m \leq M$. For a Bermudan call option with strike $K$, the equation $\tilde{V}_m(K_m^+) = K_m^+ - K$ has at most one finite solution, and
\[
  0 < \tilde{V}_m(S_2) - \tilde{V}_m(S_1) < S_2-S_1,\qquad \forall 1 \leq m \leq M-1,\ \forall S_2 > S_1 > 0.
\]
For a Bermudan put option with strike $K$, the equation $\tilde{V}_m(K_m^-) = K - K_m^-$ has at most one finite solution, and
\[
  0 < \tilde{V}_m(S_1) - \tilde{V}_m(S_2) < S_2-S_1,\qquad \forall 1 \leq m \leq M-1,\ \forall S_2 > S_1 > 0.
\]
\end{prop}
\begin{proof}
Clearly $Q = 0$ since $q_m \geq 0$. The proofs for call and put options are very similar, and we will present the argument for put options only which proceeds by induction. Since $\tilde{V}_{M-1}$ is the value of a European vanilla put option and $q_m \geq 0$, its delta is between $-1$ and 0 \cite{book}. Thus
\[
  0 < \tilde{V}_{M-1}(S_1) - \tilde{V}_{M-1}(S_2) = -\int_{S_1}^{S_2} \tilde{V}_{M-1}'(S)\,dS < S_2-S_1,
\]
for all $S_2 > S_1 > 0$. Now suppose $\tilde{V}_{M-1}(S_3) = K - S_3$ and $\tilde{V}_{M-1}(S_4) = K - S_4$ for some $S_4 > S_3 > 0$. This means
\[
  S_4-S_3 = \tilde{V}_{M-1}(S_3) - \tilde{V}_{M-1}(S_4) < S_4-S_3,
\]
which is a contradiction. So the proposition is true for $m = M-1$.

Suppose now the results hold for some $2 \leq m \leq M-1$, which implies that the function $S-K + \tilde{V}_m(S)$ is strictly increasing. Since $\tilde{V}_m(S) > 0$, for sufficiently large $S$ we have $S-K + \tilde{V}_m(S) > 0$, or $\tilde{V}_m(S) > K-S$. This means that
\[
  K_m^- = \inf \bigl\{ S > 0\colon \tilde{V}_m(S) > K-S \bigr\}
\]
is well-defined and satisfies $K_m^- < \infty$. Consider now (cf. \eqref{vtv0})
\begin{equation}
  \label{vtv1}
  V_m(S) =
  \begin{cases}
    \tilde{V}_m(S), & S \geq K_m^- \\
    K-S, & S < K_m^-
  \end{cases},
\end{equation}
and any $S_2 > S_1 > 0$. If $S_2 < K_m^-$ (and hence $S_1 < K_m^-$), we have
\[
  V_m(S_1) - V_m(S_2) = (K-S_1) - (K-S_2) = S_2-S_1,
\]
by \eqref{vtv1}. If $S_1 \geq K_m^-$ (and hence $S_2 \geq K_m^-$), we have
\[
  V_m(S_1) - V_m(S_2) = \tilde{V}_m(S_1) - \tilde{V}_m(S_2) \in (0,S_2-S_1),
\]
by \eqref{vtv1} and inductive hypothesis. If $S_1 < K_m^- \leq S_2$, we have
\begin{align*}
  V_m(S_1) - V_m(S_2) & = (K-S_1) - \tilde{V}_m(S_2) \\
  & < (K-S_1) - (K-S_2) = S_2-S_1, \\
  V_m(S_1) - V_m(S_2) & > (K-S_1) - \tilde{V}_m(S_1) \geq 0,
\end{align*}
by \eqref{vtv1}, inductive hypothesis, and the definition of $K_m^-$. In conclusion, we have shown that
\begin{subequations}\label{vv}
\begin{align}
  \label{vv1}
  0 < V_m(S_1) - V_m(S_2) = S_2-S_1,&\qquad \forall K_m^- > S_2 > S_1 > 0, \\
  \label{vv2}
  0 < V_m(S_1) - V_m(S_2) < S_2-S_1,&\qquad \text{otherwise}.
\end{align}
\end{subequations}

With the aid of \eqref{exp1} and \eqref{vv}, we write
\[
  \tilde{V}_{m-1}(S_1) - \tilde{V}_{m-1}(S_2) = e^{-2r_m \tau_m/\sigma_m^2} (I_1+I_2),
\]
where
\begin{align*}
  0 < I_1 & = \int_0^{K_m^-/S_2} \bigl[ V_m(S_1 z) - V_m(S_2 z) \bigr] \rho_m(z,1)\,dz \\
  & = (S_2-S_1) \int_0^{K_m^-/S_2} z\rho_m(z,1)\,dz, \\
  0 < I_2 & = \int_{K_m^-/S_2}^\infty \bigl[ V_m(S_1 z) - V_m(S_2 z) \bigr] \rho_m(z,1)\,dz \\
  & < (S_2-S_1) \int_{K_m^-/S_2}^\infty z\rho_m(z,1)\,dz.
\end{align*}
As a result,
\begin{align}
  \label{s1s2}
  0 & < \tilde{V}_{m-1}(S_1) - \tilde{V}_{m-1}(S_2) \\
  & < e^{-2r_m \tau_m/\sigma_m^2} (S_2-S_1) \int_0^\infty z\rho_m(z,1)\,dz \leq S_2-S_1, \nonumber
\end{align}
since by elementary properties of lognormal distributions \cite{log},
\[
  e^{-2r_m \tau_m/\sigma_m^2} \int_0^\infty z\rho_m(z,1)\,dz = e^{-2q_m \tau_m/\sigma_m^2} \leq 1.
\]
Now suppose $\tilde{V}_{m-1}(S_3) = K - S_3$ and $\tilde{V}_{m-1}(S_4) = K - S_4$. This means
\[
  \tilde{V}_{m-1}(S_3) - \tilde{V}_{m-1}(S_4) = S_4-S_3,
\]
so by \eqref{s1s2} we must have $S_3 = S_4$. The proposition then follows from induction.
\end{proof}

\begin{cor}
\label{mono}
Assume a Bermudan put option with strike $K$ has an optimal early-exercise level $K_m^- > 0$ for some $1 \leq m \leq M-1$. Then we have $\tilde{V}_m(S) > K-S$ for $S > K_m^-$ and $\tilde{V}_m(S) < K-S$ for $S < K_m^-$. Similarly, for a Bermudan call option we have $\tilde{V}_m(S) < S-K$ for $S > K_m^+$ and $\tilde{V}_m(S) > S-K$ for $S < K_m^+$.
\end{cor}
\begin{proof}
We will present the proof for put options only as the argument for call options is similar. It follows from Proposition \ref{bermb} that the function $\tilde{V}_m(S) + S-K$ is increasing in $S$. Since $\tilde{V}_m(K_m^-) + K_m^- - K = 0$, we have $\tilde{V}_m(S) > K-S$ for $S > K_m^-$ and $\tilde{V}_m(S) < K-S$ for $S < K_m^-$.
\end{proof}

\section*{Acknowledgements}

We are very grateful to Qingshuo Song for his valuable insights and helpful suggestions, as well as to Zhenan Sui for her careful reading and editing of the manuscript.

\bigskip

\centerline{\scshape Min Huang}
\medskip
{\footnotesize
  % please put the address of the first author
  \centerline{China Merchants Bank}
  \centerline{7088 Shennan Boulevard, Shenzhen, Guangdong, China}
  \centerline{\email{Email: huang.479@osu.edu}}
}

\bigskip

\centerline{\scshape Guo Luo}
\medskip
{\footnotesize
  % please put the address of the first author
  \centerline{Department of Mathematics, City University of Hong Kong}
  \centerline{Tat Chee Avenue, Kowloon, Hong Kong}
  \centerline{\email{Email: guoluo@cityu.edu.hk}}
}
%
%\bigskip
%
%\centerline{\scshape Zhenan Sui}
%\medskip
%{\footnotesize
% % please put the address of the first author
% \centerline{Institute for Advanced Study in Mathematics, Harbin Institute of Technology}
% \centerline{92 West Dazhi Street, Harbin, Heilongjiang, China}
% \centerline{\email{ Email: sui.4@osu.edu}}
%}

\end{document}